\documentclass{amsart}
\usepackage{amsfonts}
\usepackage{amsmath}
\usepackage{amssymb}
\usepackage{graphicx}
\usepackage{enumerate}
\usepackage{lineno}

\newtheorem{theorem}{Theorem}
\theoremstyle{plain}

\newtheorem{case}{Case}

\newtheorem{definition}{Definition}
\newtheorem{example}{Example}

\newtheorem{lemma}{Lemma}

\newtheorem{remark}{Remark}

\numberwithin{equation}{section}

\begin{document}
\title[]{Some Classes of Non-archimedean Radial Probability
Density Functions Associated With Energy Landscapes}

\begin{abstract}
In this article, we study a large class of radial probability density functions defined on the $p$-adic numbers
from which it is possible to obtain certain non-archimedean pseudo-differential operators. These operators are associated with certain $p$-adic master equations of some models of complex systems (such as glasses, macromolecules, and proteins). We prove via the theory of distributions some properties corresponding to the heat Kernel associated with these pseudo-differential operators. Also,
study some properties corresponding to the fundamental solution of these $p$-adic equations. Finally, we will study strong Markov processes, the first passage time problem and the survival probability (of the trajectories of these processes) corresponding to radial probability density functions connected with energy landscapes of the linear and logarithmic types.
\end{abstract}

\author{Ismael Guti\'{e}rrez-Garc\'{\i}a}
\email{isgutier@uninorte.edu.co}
\author{Anselmo Torresblanca-Badillo}
\email{atorresblanca@uninorte.edu.co}
\address{Universidad del Norte, Departamento de Matem\'{a}ticas y Est\'{a}%
distica, Km. 5 V\'{\i}a Puerto Colombia. Barranquilla, Colombia.}
\maketitle

\section{Introduction}
Many problems in physics as the dynamics of complex systems are described by a random walk on a complex energy landscape, see e.g. \cite{Fraunfelder et al 2}-\cite{Kozyrev  SV}, \cite{Me-Pa-Pa}-\cite{S-W-2}, \cite{To-z-2}-\cite{Zu-2}.\ An energy landscape (or simply a landscape) is a mapping of all possible conformations of a molecular entity, or the spatial positions of interacting molecules in a system. Mathematically, an energy landscape is a continuous function $\mathbb{U} : X \rightarrow \mathbb{R}$ that assigns to each physical state of a system its energy, where $X$ is a topological space. The term \textit{complex landscape} means that the function $\mathbb{U}$ has many local minima.

Along this article, $p$ will denote a prime number. The field of
$p$-adic numbers $\mathbb{Q}_p$ is defined as the completion of the field of rational numbers $\mathbb{Q}$ concerning the $p$-adic norm $|\cdot |_p$, which is defined as
\begin{equation*}
|x|_p = \begin{cases}
0, & \text{if}  \ \ x=0 \\
p^{-\gamma }, & \text{if} \ \  x=p^{\gamma }\frac{a}{b}\text{,}
\end{cases}
\end{equation*}
where $a$ and $b$ are integers coprime with $p$. The integer $\gamma :=ord_p(x) $, with $ord_p(0):=+\infty $, is called the\textit{\ }$p$-\textit{adic order of} $x$.

We extend the $p$-adic norm to $\mathbb{Q}_p^n$ by taking
\begin{equation*}
\|x\|_p:=\max_{1\leq i\leq n}|x_i|_p,\text{ for } x=(x_1,\ldots ,x_n)\in \mathbb{Q}_p^n.
\end{equation*}
In \cite{Av-4}-\cite{Av-5} Avetisov et al. developed, in dimension one, new class of models associated with $p$-adic
pseudodifferential operators which describe the dynamics of many complex systems. In these models, the time-evolution of the
system is controlled by a master equation of the form
\begin{equation}\label{Master_E}
\frac{\partial u(x,t)}{\partial t} = \int\limits_{\mathbb{Q}_p} \{ j (x\mid y) \ u(y,t)-j(y\mid x)\ u(x,t)\} dy, x \in \mathbb{Q}_p, t\in \mathbb{R}_{+},
\end{equation}
where the function $u\left( x,t\right) :\mathbb{Q}_p\times \mathbb{R}_{+}\rightarrow \mathbb{R}_{+}$ is a probability
density distribution, and the function $j\left( x\mid y\right) :\mathbb{Q}_p\times \mathbb{Q}_p\rightarrow \mathbb{R}_{+}$ is the probability of transition from state $y$ to the state $x$ per unit of time. The transition from a state $y$ to a state $x$ can be perceived as overcoming the energy barrier \ separating these states. In \cite{Av-4} an Arrhenius type relation was used, that is,
\begin{equation*}
j(x\mid y) \sim A(T)\exp \left\{ -\frac{\mathbb{U} (x\mid y)}{kT} \right\},
\end{equation*}
where $\mathbb{U}\left( x\mid y\right) $ is the height of the activation barrier for the
transition from the state $y$ to state $x$, $k$ is the Boltzmann constant, $T$ is the temperature
and $A(T)$ is a positive constant depending  of $T$. This formula establishes a
relation between the structure of the energy landscape $\mathbb{U} (x\mid y)$ and the transition function $j(x\mid y)$. The
case $j(x\mid y) =j(y\mid x)$ corresponds to a \textit{degenerate energy landscape}. In this case the master equation (\ref{Master_E}) takes the form
\begin{equation*}\label{Master_E1}
\frac{\partial u(x,t)}{\partial t} = \int\limits_{\mathbb{Q}_p} j\left(|x-y|_p\right) \{u(y,t)-u(x,t)\} dy,
\end{equation*}
where $j\left(|x-y|_p\right) =\frac{A(T)}{|x-y|_p}\exp \left\{ -\frac{\mathbb{U}(|x-y|_p) }{kT}\right\} $.

Relaxation kinetics in complex systems is often fitted into three empirical laws, namely, the stretched exponential curve (the Kohlrausch-Williams-Watts law), $\sim e^{-(\frac{t}{\tau})^{\alpha}}$, $0<\alpha<1$, the power decay law,
$\sim \left(\frac{t}{\tau}\right)^{-\alpha}$, $0<\alpha$, and the logarithmic decay law, $\sim \alpha\left[\ln\left(\frac{t}{\tau}\right)\right]^{-1}$, $1<\alpha$.
These types of relaxation are characteristic of complex systems, see \cite{Av-4}, \cite{Brawer}, \cite{Matsuoka} and \cite{Ngai}. Motivated by this, is that three hierarchical energy landscapes are considered: logarithmic, linear and exponential landscapes.

There are three probability of transition density functions $j$ associated respectively with each of the energy landscapes described above. These transition functions are:
\begin{enumerate}[(i)]
\item (logarithmic) $j\left(|x-y|_p\right) = \frac{1}{|x-y|_p \ln ^{\alpha }(1+|x-y|_p) }$, $\alpha>1$,
\item (linear) $j\left(|x-y|_p\right) =\frac{1}{|x-y|_p^{\alpha +1}}$, $\alpha >0$,
\item (exponential) $j\left(|x-y|_p\right) =\frac{e^{-\alpha |x-y|_p}}{|x-y|_p}$, $\alpha >0$.
\end{enumerate}

The equations of the form
\begin{equation}
u_t(x,t)=(J\ast u-u)(x,t) = \int_{\mathbb{Q}_p^n} \hskip-0.22cm J(x-y) u(y,t) d^ny - u(x,t), \ x\in\mathbb{Q}_p^n, \ t\in [0,\infty), \label{Master_E2}
\end{equation}
and its connections with Stochastic processes on ultrametric spaces were recently studied in arbitrary dimensions, where the functions $J$ codify the structure of the transition density functions from the energy landscapes of the exponential type, see \cite{To-Z}.

As in the Archimedean case, as stated in \cite{Andreu-Vaillo et al} and \cite{Fife}, if $u(x,t)$ is thought of as a density at a point $x$ at time $t$ and $J(x-y)$ is thought of as the probability distribution of jumping from location $y$ to location $x$, then
\begin{equation*}
\int_{\mathbb{Q}_p^n} J(x-y)u(y,t)d^ny = (J\ast u)(x,t)
\end{equation*}
is the rate at which individuals are arriving at position $x$ from all other places and
\begin{equation*}
-u(x,t) = -\int_{\mathbb{Q}_p^n} J(x-y) u(y,t) d^ny
\end{equation*}
is the rate at which they are leaving location $x$ to travel to all other sites.

In this article we continue the study of the master equation (\ref{Master_E2}) naturally
associated to nonlocal operators $\mathcal{A}$\ defined as $\mathcal{A}
f:=J\ast f-f,$ with $f\in L^{\rho }\left(\mathbb{Q}_p^n\right) $, $1\leq \rho \leq \infty$.

An important fact that makes a huge difference between this article and the works \cite{Casas-Zuniga}, \cite{Ch-Z-1},
\cite{Galeano-Zuniga}, \cite{Gu-To-1}, \cite{Gu-To-2}, \cite{Zu}, et al., is that the heat Kernel $Z_t(x)$, $x\in\mathbb{Q}_p^n$, $t\geq 0$ associated
with the operator $\mathcal{A}$ corresponds to a $p$-adic distribution.

The article is organized as follows: In Section \ref{Fourier Analysis}, we will collect some basic results on the $p$-adic analysis and fix the notation that we will use through the article. In Section \ref{nonlocal operators}, we started showing some general results corresponding to radial probability density
functions $J:\mathbb{Q}_p^n\rightarrow \mathbb{R}_{+}$ that satisfies the Hypothesis A (see Definition \ref{Hypothesis A}). Subsequently, we will define a class of $p$-adic operators $\mathcal{A}$ associated with the function $J$.
These operators are pseudo-differential operators with symbol $1-(\mathcal{F}J)$, where $\mathcal{F}J$ denotes the Fourier transform of $J$. Some properties corresponding to the symbol $1-(\mathcal{F}J)$ will be shown. Moreover, throughout this paper, we will assume that $1-(\mathcal{F}J)(\|\xi\|_p)$ is a increasing function with respect to $\|\cdot\|_p$.

We also study the Cauchy problem
\begin{equation}\label{Cauchy_problem_0}
\begin{cases}
\frac{\partial u}{\partial t}(x,t)=\mathcal{A}u(x,t), & t\in
[0,\infty), \ x\in\mathbb{Q}_p^n \\
u(x,0)=u_0(x)
\end{cases}
\end{equation}
naturally associated to these operators, and we prove via the theory of distributions
some properties corresponding to the heat Kernel associated with the operators $\mathcal{A}$, see Theorem \ref{Z(t,x)_properties}. Also study some properties corresponding to the fundamental solution $u(x,t)$, $x\in \mathbb{Q}_p^n$, $t\geq 0$, of the Cauchy problem (\ref{Cauchy_problem_0}),
such as, the mass of the solution $u(x,t)$ is conserved, the comparison principle holds and some representations of $u(x,t)$ under certain conditions of the initial data $u_0(x)$. see Theorem \ref{theorem_classical_solution}.

Finally, in this section we will show that the family $\left(Z(t,\cdot)\right)_{t>0}$ determines a convolution
semigroup on $\mathbb{Q}_p^n$, see Theorem \ref{conv_semigroups}.
In Section \ref{Markov Processes}, we will begin by defining radial probability density functions of the linear and logarithmic types. These functions are defined on $\mathbb{Q}_p^n$. Then, we look conditions for that the Markov process $\mathfrak{J}(t,\omega )$ (with state space $(\mathbb{Q}_p^n,\|\cdot\|_p)$) associated with these transition functions are recurrent concerning  $\mathbb{Z}_p^n$ (The first passage time problem), see Theorem \ref{Theorem3}. The first passage time problem for the transition functions of the type exponential was studied in \cite{To-Z}.

Ultimately, we will study the probability of survival for the transition functions of the types linear and logarithmic.

\section{\label{Fourier Analysis}Fourier Analysis on $\mathbb{Q}_p^n$: Essential Ideas}

\subsection{Non-Archimedean topology of the $p$-adic numbers}

Any $p$-adic number $x\neq 0$ has a unique expansion of the form
\begin{equation*}
x=p^{ord_p(x)}\sum_{j=0}^{\infty }x_jp^j,
\end{equation*}
where $x_j\in \{0,1,2,\dots ,p-1\}$ and $x_0\neq 0$. By using this
expansion, we define \textit{the fractional part of }$x\in \mathbb{Q}_p$,
denoted $\{x\}_p$, as the rational number
\begin{equation*}
\{x\}_p =   \begin{cases}
0, & \text{if} \ x=0 \ \text{or} \ ord_p(x)\geq 0 \\
p^{ord_p(x)} \sum_{j=0}^{-ord_p(x)-1} x_jp^j, & \text{if} \     ord_p(x)<0.
    \end{cases}
\end{equation*}
For $r\in \mathbb{Z}$, denote by $B_r^n(a) = \{x \in \mathbb{Q}_p^n: \|x-a\|_p\leq p^r\}$
\textit{the ball of radius} $p^r$ \textit{with center at}
$a =(a_1,\ldots,a_n)\in \mathbb{Q}_p^n$, and take $B_r^n(0)=: B_r^n$. Note that $B_r^n(a) = B_r(a_1)\times \cdots \times B_r(a_n)$, where $B_r(a_{i}):=\{x\in \mathbb{Q}_p :  |x_i -a_i|_p \leq p^r\}$ is the one-dimensional ball of radius $p^r$ with center at $a_i \in \mathbb{Q}_p$. The ball $B_0^n$ equals the product of $n$ copies of $B_0=\mathbb{Z}_p$, \textit{the ring of
}$p$-\textit{adic integers of} $\mathbb{Q}_p$. We also denote by
$S_r^n(a)=\{x\in \mathbb{Q}_p^n : \|x-a\|_p=p^r\}$ \textit{the sphere of radius }$p^r$ \textit{with center at} $a=(a_1,\dots ,a_n)\in \mathbb{Q}_p^n$, and take $S_r^n(0)=:S_r^n$. The balls and spheres are both open and closed subsets in $\mathbb{Q}_p^n$.
As a topological space, $(\mathbb{Q}_p^n,\|\cdot \|_p)$ is totally disconnected, i.e., the only connected subsets of $\mathbb{Q}_p^n$ are the empty set and the points. A subset of $\mathbb{Q}_p^n$ is compact if and only if it is closed and bounded in $\mathbb{Q}_p^n$, see, e.g., \cite[Section 1.3]{V-V-Z}, or \cite[Section 1.8]{Alberio et al}. The balls and spheres are compact subsets. Thus $(\mathbb{Q}_p^n,\|\cdot \|_p) $ is a locally compact topological space.

We will use $\Omega \left( p^{-r}\|x-a\|_p\right)$ to denote the
characteristic function of the ball $B_r^n(a)$. We will use the notation $1_{A}$ for the characteristic function of a set $A$. Along the article $d^nx$ will denote a Haar measure on $\mathbb{Q}_p^n$ normalized so that $\int_{\mathbb{Z}_p^n}d^nx=1.$

\subsection{Some function spaces}

A complex-valued function $\varphi $ defined on $\mathbb{Q}_p^n$ is called \textit{locally constant} if for any $x\in\mathbb{Q}_p^n$ there
exist an integer $l(x)\in \mathbb{Z}$ such that
\begin{equation*}
\varphi (x+x') = \varphi(x) \text{ for } x'\in B_{l(x)}^n.
\end{equation*}
A function $\varphi :\mathbb{Q}_p^n\rightarrow \mathbb{C}$ is called a \textit{Bruhat-Schwartz function
(or a test function)} if it is locally constant with compact support. The $\mathbb{C}$-vector space of Bruhat-Schwartz functions is denoted by $
\mathcal{D}(\mathbb{Q}_p^n)=: \mathcal{D}$. Let $\mathcal{D}'(\mathbb{Q}_p^n)=: \mathcal{D}'$ denote the set of all continuous functional (distributions) on $\mathcal{D}$. The natural pairing
$\mathcal{D}' (\mathbb{Q}_p^n)\times
\mathcal{D}(\mathbb{Q}_p^n)\rightarrow \mathbb{C}$ is denoted as
$\langle T,\varphi \rangle $ for $T\in \mathcal{D}' (\mathbb{Q}_p^n)$ and $\varphi \in
\mathcal{D}(\mathbb{Q}_p^n)$, see e.g., \cite[Section
4.4]{Alberio et al}.

Every $f\in $ $L_{loc}^1(\mathbb{Q}_p^n)$ defines a
distribution $f\in \mathcal{D}' (\mathbb{Q}_p^n)$ by the formula
\begin{equation*}
\langle f,\varphi \rangle =\int\limits_{\mathbb{Q}_p^n}f(x) \varphi(x) d^nx.
\end{equation*}%
Such distributions are called \textit{regular distributions}.

Given $\rho \in \lbrack 0,\infty )$, we denote by $L^{\rho
}\left(\mathbb{Q}_p^n,d^nx\right) =L^{\rho
}\left(\mathbb{Q}_p^n\right) :=L^{\rho },$ the
$\mathbb{C}$-vector space of all the complex valued functions $g$
satisfying $\int_{\mathbb{Q}_p^n}\left\vert g\left( x\right)
\right\vert ^{\rho }d^nx<\infty $, $L^{\infty }\allowbreak
:=L^{\infty }\left(\mathbb{Q}_p^n\right) =L^{\infty
}\left(\mathbb{Q}_p^n,d^nx\right) $ denotes the
$\mathbb{C}$-vector space of all the complex valued functions $g$
such that the essential supremum of $|g|$ is bounded.

Let us denote by $C(\mathbb{Q}_p^n,\mathbb{C})=:C_{\mathbb{C}}$
the $\mathbb{C}$-vector space of all the complex valued functions which are continuous, by $
C(\mathbb{Q}_p^n,\mathbb{R})=:C_{\mathbb{R}}$ the
$\mathbb{R}$-vector space of continuous functions. Set
\begin{equation*}
C_0(\mathbb{Q}_p^n,\mathbb{C}):=\left\{
f : \mathbb{Q}_p^n\rightarrow \mathbb{C} ; f\text{ is
continuous and }\lim_{x\rightarrow \infty }f(x)=0\right\} ,
\end{equation*}%
where $\lim_{x\rightarrow \infty }f(x)=0$ means that for every $\epsilon >0$
there exists a compact subset $B(\epsilon )$ such that $|f(x)|<\epsilon $
for $x\in\mathbb{Q}_p^n\backslash B(\epsilon ).$ We recall that $(C_0(\mathbb{Q}_p^n,\mathbb{C}),\|\cdot \|_{L^{\infty }})$ is a Banach space.

\subsection{Fourier transform}
Set $\chi _p(y)=\exp (2\pi i\{y\}_p)$ for $y\in\mathbb{Q}_p$. The map $\chi _p(\cdot )$ is an additive character on $\mathbb{Q}_p$, i.e. a
continuous map from $\left(\mathbb{Q}_p,+\right) $ into $S$ (the unit circle considered as multiplicative
group) satisfying $\chi _p(x_0+x_1)=\chi _p(x_0)\chi _p(x_1)$,
$\left|\chi _p(x_0) \right|=1$, $x_0,x_1\in\mathbb{Q}_p$. Further, $\chi _p(0)=1$, $\chi _p(-x)=\overline{\chi _p(x)}$.
Here, $\overline{\chi _p}$ denotes the complex conjugate of $\chi _p$. \\
The additive characters of $\mathbb{Q}_p$ form an Abelian group which is isomorphic to $\left(\mathbb{Q}_p,+\right) $, the
isomorphism is given by $\xi \mapsto \chi _p(\xi x)$,
see e.g. \cite[Section 2.3]{Alberio et al}.

Given $x=(x_1,\dots,x_n),$ $\xi =(\xi_1,\dots ,\xi_n) \in\mathbb{Q}_p^n$, we set $x\cdot \xi :=\sum_{j=1}^nx_j\xi_j$. If $f\in L^1$ its Fourier transform is denoted and defined by
\begin{equation*}
(\mathcal{F}f)(\xi) = \int_{\mathbb{Q}_p^n}\chi _p(\xi \cdot x)f(x)d^nx,\quad \text{for }\xi \in
\mathbb{Q}_p^n.
\end{equation*}
The inverse Fourier transform is denoted by $\mathcal{F}^{-1}$.
The Fourier transform is a linear isomorphism from $\mathcal{D}(\mathbb{Q}_p^n)$ onto itself satisfying
\begin{equation}
(\mathcal{F}(\mathcal{F}f))(\xi )=f(-\xi ),  \label{FF(f)}
\end{equation}
for every $f\in \mathcal{D}(\mathbb{Q}_p^n),$ see e.g.
\cite[Section 4.8]{Alberio et al}. If $f\in L^2,$ its Fourier
transform is defined as
\begin{equation*}
(\mathcal{F}f)(\xi )=\lim_{k\rightarrow \infty }\int_{\|x\|\leq
p^{k}}\chi _p(\xi \cdot x)f(x)d^nx,\quad \text{for }\xi \in
\mathbb{Q}_p^n,
\end{equation*}
where the limit is taken in $L^2.$ We recall that the Fourier transform is
unitary on $L^2$, i.e. $\|f\|_{L^2}=\|\mathcal{F}f\|_{L^2}$ for $f\in
L^2$ and that (\ref{FF(f)}) is also valid in $L^2$, see e.g. \cite[Chapter III, Section 2]{Taibleson}.

\section{\label{nonlocal operators}A Class of Nonlocal p-adic Operators and its Heat Kernel.}

\begin{definition}[Hypothesis A]
\label{Hypothesis A} We say that a function
$J:\mathbb{Q}_p^n\rightarrow \mathbb{R}_{+}$ satisfies
Hypothesis A if the following hypotheses are met, which will be
assumed throughout this paper:
\begin{enumerate}[(i)]
\item $J$ is a radial function, i.e. $J(x)=J(\|x\|_p)$;
\item $J$ is a continuous function (a.e.) with
$\int_{\mathbb{Q}_p^n}J(\|x\|_p)d^nx=1$.
\end{enumerate}
\end{definition}
The above assumptions imply that $J$ is a radial probability
density.
\begin{lemma}\label{Lemma_J}
\label{lemma1}Assume that $J:$\textbf{\
}$\mathbb{Q}_p^n\rightarrow \mathbb{R}_{+}$ satisfies Hypothesis
A$.$ Then, the following assertions hold:
\begin{enumerate}[(i)]
\item $(\mathcal{F}J)(\xi )$ is a real-valued, radial (i.e. $(\mathcal{F}J)(\xi )= (\mathcal{F}J)(\|\xi \|_p)$), and continuous function,
satisfying $ (\mathcal{F}J)(0)=1$ and $|(\mathcal{F}J)(\|\xi
\|_p)| \leq 1$, for all $\xi \in \mathbb{Q}_p^n$;
\item For $\xi \in \mathbb{Q}_p^n\backslash \{0\}$,
\begin{equation*}
1-(\mathcal{F}J)(\|\xi \|_p) =\|\xi \|_p^{-n}\left\{J(p\|\xi
\|_p^{-1})p^n+(1-p^{-n})\sum_{j=2}^{\infty
}J(p^j\|\xi \|_p^{-1})p^{nj}\right\}.
\end{equation*}
\end{enumerate}
\end{lemma}

\begin{proof}
\begin{enumerate}[(i)]
\item The first two statements are obtained by applying the $n$-dimensional
version in \cite[Example 8, p. 43]{V-V-Z}. The continuity of $(\mathcal{F}J)(\|\xi \|_p)$ follows from\ \cite[Chapter II-1-Theorem 1.1-(b)]{Taibleson}. The
affirmations
$|(\mathcal{F}J)(\|\xi \|_p)|\leq 1$ and $(\mathcal{F}J)(0) =1$, are a direct consequence of the fact that $\int_{\mathbb{Q}_p^n}J(\|x\|_p)d^nx=1$.

\item Let $\xi=p^{L}\xi_0\neq 0$, with $L=ord_p(\xi)\in \mathbb{Z}$ and $\|\xi_0\|_p=1$. Making the change of variable $z=p^{L}x$ we have that
\begin{align*}
1-(\mathcal{F}J)(\|\xi\|_p)&=\int_{\mathbf{\mathbb{Q}}_p^n}J(\|x\|_p)\left\{1-\chi_p\left(
p^{L}x\cdot \xi_0\right)\right\}d^nx\\
&=\|\xi \|_p^{-n}\int_{\mathbf{\mathbb{Q}}_p^n}J(\|\xi\|_p^{-1}\|z\|_p)\left\{1-\chi_p\left(
z\cdot \xi_0\right)\right\}d^nz\\
&=\|\xi \|_p^{-n}\int_{\mathbf{\mathbb{Q}}_p^n\setminus \mathbf{\mathbb{Z}}_p^n}J(\|\xi\|_p^{-1}\|z\|_p)\left\{1-\chi_p\left(
z\cdot \xi_0\right)\right\}d^nz\\
&=\|\xi \|_p^{-n}\sum_{j=1}^{\infty}J(\|\xi\|_p^{-1}p^j)\int_{\|z\|_p=p^j}\left\{1-\chi_p\left(
z\cdot \xi_0\right)\right\}d^nz\\
&=\|\xi \|_p^{-n}\sum_{j=1}^{\infty}J(\|\xi\|_p^{-1}p^j)\int_{\|p^jz\|_p=1}\left\{1-\chi_p\left(
z\cdot \xi_0\right)\right\}d^nz\\
&=\|\xi \|_p^{-n}\sum_{j=1}^{\infty}J(\|\xi\|_p^{-1}p^j)p^{nj}\int_{\|w\|_p=1}\left\{1-\chi_p\left(
p^{-j}\xi_0\cdot w\right)\right\}d^nw.
\end{align*}
By using the formula
\begin{equation*}
\int_{\|w\|_p=1}\chi_p\left(p^{-j}\xi_0\cdot w\right)
d^nw=\left\{
\begin{array}{lll}
1-p^{-n}, & \text{if} & \text{ }j\leq 0, \\
-p^{-n}, & \text{if } & \text{\ }j=1, \\
0, & \text{if} & j\geq 2,
\end{array}
\right.
\end{equation*}
we have that
\begin{equation*}
\int_{\|w\|_p=1}\left\{1-\chi_p\left(p^{-j}\xi_0\cdot w\right)\right\}
d^nw=\left\{
\begin{array}{lll}
0, & \text{if} & \text{ }j\leq 0, \\
1, & \text{if } & \text{\ }j=1, \\
1-p^{-n}, & \text{if} & j\geq 2.
\end{array}
\right.
\end{equation*}
Therefore,
\begin{equation*}
1-(\mathcal{F}J)(\|\xi \|_p) =\|\xi \|_p^{-n}\left\{J(p\|\xi
\|_p^{-1})p^n+(1-p^{-n})\sum_{j=2}^{\infty
}J(p^j\|\xi \|_p^{-1})p^{nj}\right\}.
\end{equation*}
\end{enumerate}
\end{proof}

\begin{remark}\label{obs_J}
\begin{enumerate}[(i)]
\item By previous lemma we have that $|(\mathcal{F}J)(\|\xi
\|_p)| \leq 1$. This fact implies that $0\leq 1-
(\mathcal{F}J)(\|\xi \|_p)\leq 2$, for all $\xi\in
\mathbb{Q}_p^n$.
\item A function $f:\mathbb{Q}_p^n\rightarrow  \mathbb{C}$ is called negative definite, if$\ $
\begin{equation*}
\sum_{i,j=1}^{m}\left( f(x_{i})+\overline{f(x_j)}
-f(x_{i}-x_j)\right) \lambda _{i}\overline{\lambda _j}\geq 0
\end{equation*}%
for all $x_1,\ldots ,x_{m}\in $ $\mathbb{Q}_p^n,$ $\lambda
_1,\ldots ,\lambda _{m}$ $\in $ $\mathbb{C}$, $m\in \mathbb{N}$.
We have that the function $(\mathcal{F}J)(0) - (\mathcal{F}J)(\|\xi
\|_p) = 1 - (\mathcal{F}J)(\|\xi \|_p)$ is negative definite,
see e.g. \cite[Remark 3-(ii)]{To-Z} and \cite[Example 3.4]{To-z-2}.
\end{enumerate}
\end{remark}

Throughout this paper (unless otherwise stated) we will assume that
$1-(\mathcal{F}J)(\|\xi \|_p)$ is a increasing function with
respect to $\|\cdot\|_p$. In the section \ref{Markov Processes} we will show examples of transition functions $J$ for
which $1-(\mathcal{F}J)(\|\xi \|_p)$ is a increasing function with respect to $\|\cdot\|_p$.

\begin{remark}\label{Obs_operator}\cite[Section 4.2]{To-Z}
For $f\in L^{\rho }\left(\mathbb{Q}_p^n\right) $ with $1\leq
\rho \leq \infty ,$\ we define $$\mathcal{A}f:=J\ast f-f.$$ Then,
for any $1\leq \rho \leq \infty $, $\mathcal{A}:L^{\rho
}\longrightarrow L^{\rho }$ gives rise a well-defined linear bounded
operator. Indeed, by the Young inequality
\begin{equation*}
\|\mathcal{A}f\|_{L^\rho}\leq \|J\ast f\|_{L^\rho} + \|f\|_{L^\rho} \leq \|J\|_{L^1}\|f\|_{L^\rho}+\|f\|_{L^{\rho }}\leq 2\|f\|_{L^\rho}.
\end{equation*}
Consider $\mathcal{A}:L^2(\mathbb{Q}_p^n) \rightarrow L^2(\mathbb{Q}_p^n)$ given by
\begin{equation*}
\mathcal{A}f(x) =-\mathcal{F}^{-1} ((1- (\mathcal{F}J)(\|\xi\|_p)) \mathcal{F}f) ,
\end{equation*}
and the Cauchy problem:
\begin{equation}
\begin{cases}
\frac{\partial u}{\partial t}(x,t)=\mathcal{A}u(x,t), & t\in
\left[ 0,\infty \right) \text{,\ }x\in\mathbb{Q}_p^n \\
u(x,0)=u_0(x)\in \mathcal{D}(\mathbb{Q}_p^n)\text{.} &
\end{cases}
 \label{Cauchy_problem}
\end{equation}

Then
\begin{equation}
u(x,t)=\int_{\mathbf{\mathbb{Q}}_p^n}\chi_p\left( -\xi \cdot
x\right) e^{-t(1-(\mathcal{F}J)(\|\xi \|_p))}
(\mathcal{F}u_0)(\xi )d^n\xi \label{classical_solution}
\end{equation}
is a classical solution of \eqref{Cauchy_problem}. In addition,
$u(\cdot ,t)$ is a continuous function for any $t\geq 0$.
\end{remark}

We define the heat Kernel (or classical solution) attached to
operator $\mathcal{A}$ as
\begin{equation}\label{def_Z(x,t)}
Z(x,t):=\left(\mathcal{F}^{-1}e^{-t(1-(\mathcal{F}J)(\|\xi
\|_p))}\right)(x), \ \ x\in\mathbb{Q}_p^n, \ \ t\geq 0.
\end{equation}
When $t=0$ we get that $Z(x,t)=\delta \in
\mathcal{D}'(\mathbb{Q}_p^n)$, see \cite[Example
4.9.1]{Alberio et al}. On the other hand, for $t>0$ we have that
$e^{-t(1-(\mathcal{F}J)(\|\xi \|_p))}\notin
L^1(\mathbf{\mathbb{Q}}_p^n)$ and
$e^{-t(1-(\mathcal{F}J)(\|\xi \|_p))}\in
L_{loc}^1(\mathbf{\mathbb{Q}}_p^n)$, so that by \cite[Chapter
IV]{Alberio et al} we have that $e^{-t(1-(\mathcal{F}J)(\|\xi
\|_p))}$ determine a regular distribution and
$\mathcal{F}^{-1}(e^{-t(1-(\mathcal{F}J)(\|\xi \|_p))})(x)\in
\mathcal{D}'(\mathbb{Q}_p^n)$. Therefore, we
have that $Z(x,t)\in \mathcal{D}^{\prime }(\mathbb{Q}_p^n)$, for $x\in\mathbb{Q}_p^n$ and $t\geq 0$.

When considering $Z(x,t)$ as a function of $x$ for $t$ fixed, we
will write $Z_t(x).$

By (\ref{classical_solution}) and (\ref{def_Z(x,t)}), we have that
the classical solution of the Cauchy problem (\ref{Cauchy_problem})
satisfies
\begin{equation*}
u(x,t)=Z_t(x)\ast u_0(x), \ \ x\in\mathbb{Q}_p^n, \ \ t\geq
0.
\end{equation*}
For more details the author can consult \cite{To-Z}.\\
Significantly, the following results are obtained by considering our heat Kernel $Z_t(x)$, $x\in\mathbb{Q}_p^n$, $t\geq 0$, as a distribution. This fact
makes a big difference between this article and the works \cite{Casas-Zuniga}, \cite{Ch-Z-1}, \cite{Galeano-Zuniga}, \cite{Gu-To-1}, \cite{Gu-To-2}, \cite{To-Z},
\cite{Zu}, et al.

\begin{lemma} \label{TF_Z(t,x)}
For $\xi \in \mathbb{Q}_p^n$ and $t\geq 0$ we
have that
\begin{equation*}
(\mathcal{F}Z_t(x))(\xi)=e^{-t(1-(\mathcal{F}J)(\|\xi \|_p))}.
\end{equation*}
\end{lemma}

\begin{proof}
Let $\varphi \in \mathcal{D}(\mathbb{Q}_p^n)$. Then by
\cite[Section 4.4]{Alberio et al} and using Fubini's we have that
\begin{align*}
\langle\mathcal{F}Z_t(x),\varphi\rangle &= \langle Z_t(x),\mathcal{F}\varphi \rangle\\
&=\int_{\mathbf{\mathbb{Q}}_p^n}Z_t(x)(\mathcal{F}\varphi)(x)d^nx\\
&=\int_{\mathbf{\mathbb{Q}}_p^n}\left[\int_{\mathbf{\mathbb{Q}}_p^n}\chi_p\left(
x \cdot \xi\right)Z_t(x)d^nx\right]\varphi(\xi)d^n\xi\\
&=\int_{\mathbf{\mathbb{Q}}_p^n}e^{-t(1-(\mathcal{F}J)(\|\xi
\|_p))}\varphi(\xi)d^n\xi\\
&=\left\langle e^{-t(1-(\mathcal{F}J)(\|\xi \|_p))},\varphi\right\rangle.
\end{align*}
\end{proof}

\begin{theorem} \label{Z(t,x)_properties}
The following assertions holds for $Z_t(x)$, with
$x\in\mathbb{Q}_p^n$ and $t\geq 0$:
\begin{enumerate}[(i)]
\item $Z_t(x)\geq 0$;
\item $\int_{\mathbf{\mathbb{Q}}_p^n}Z_t(x)d^nx=1$;
\item $Z_{t+s}(x)=\left(Z_t\ast Z_{s}\right)(x)$, $s\geq 0$.
\end{enumerate}
\end{theorem}

\begin{proof}
\begin{enumerate}[(i)]
\item If $x=0$ the statement is immediate. If $x \neq 0$ and $t=0$, then by \cite[Example 9, p.
44]{V-V-Z} we have that $Z_t(x)=0$.

Let $t>0$ and $x=p^{L}x_0\neq 0$, with $L=ord_p(x)$ and
$\|x_0\|_p=1$. By changing variables as $y=p^{L}\xi$, we have
that
\begin{align*}
Z_t(x) &=\int_{\mathbf{\mathbb{Q}}_p^n}\chi_p\left(
-p^{L}x_0\cdot \xi \right) e^{-t(1-(\mathcal{F}J)(\|\xi
\|_p))}d^n\xi \\
&=\|x\|_p^{-n}\int_{\mathbf{\mathbb{Q}}_p^n}\chi_p\left(-x_0\cdot y
\right)e^{-t(1-(\mathcal{F}J)(p^{L}\|y\|_p))}d^ny.
\end{align*}
Then
\begin{align*}
Z_t(x)\|x\|_p^n & = \int_{\mathbf{\mathbb{Z}}_p^n} e^{-t(1-(\mathcal{F}J)(p^{L}\|y\|_p))}d^ny + \int_{\mathbf{\mathbb{Q}}_p^n \setminus         \mathbf{\mathbb{Z}}_p^n} \hspace{-0.2cm} \chi_p\left(-x_0\cdot y \right)e^{-t(1-(\mathcal{F}J)(p^{L}\|y\|_p))}d^ny.
\end{align*}
Expanding the previous integrals, we have that $Z_t(x)\|x\|_p^n$
it's exactly
\begin{equation}
(1-p^{-n})\sum_{j=0}^{\infty}p^{-nj}e^{-t(1-(\mathcal{F}J)(p^{L-j}))}+\sum_{j=1}^{\infty}e^{-t(1-(\mathcal{F}J)(p^{L+j}))}
\hspace{-0.4cm}\underset{\|y\|_p=p^j}\int\hspace{-0.3cm}\chi_p\left(-x_0\cdot
y \right)d^ny. \label{exp Z(x,t)}
\end{equation}

By using the formula
\begin{equation*}
\int_{\|y\|_p=p^j}\chi_p\left(-x_0\cdot y\right)
d^ny=\left\{
\begin{array}{lll}
p^{nj}(1-p^{-n}), & \text{if} & \text{ }1\leq p^{-j}, \\
-p^{n(j-1)}, & \text{if } & \text{\ }1=p^{-j+1}, \\
0, & \text{if} & \ 1\geq p^{-j+2},
\end{array}
\right.
\end{equation*}
and (\ref{exp Z(x,t)}), we get that
\begin{equation}
Z_t(x)=\|x\|_p^{-n}\left\{(1-p^{-n})\sum_{j=0}^{\infty}p^{-nj}e^{-t(1-(\mathcal{F}J)(p^{-j}\|x\|_p^{-1}))}-e^{-t(1-(\mathcal{F}J)(p\|x\|_p^{-1}))}
\right\},
\end{equation} \label{exp11 Z(x,t)}
and since $(1-p^{-n})\sum_{j=0}^{\infty}p^{-nj}=1$, we have that
\begin{equation}
Z_t(x)=\|x\|_p^{-n}(1-p^{-n})\sum_{j=0}^{\infty}p^{-nj}\left\{e^{-t(1-(\mathcal{F}J)(p^{-j}\|x\|_p^{-1}))}-e^{-t(1-(\mathcal{F}J)(p\|x\|_p^{-1}))}
\right\}. \label{exp2 Z(x,t)}
\end{equation}

As $1-(\mathcal{F}J)(\|\xi \|_p)$ is a increasing function with
respect to $\|\cdot\|_p$, then for $j\in \mathbf{\mathbb{N}}_0$
we have that $1-(\mathcal{F}J)(p^{-j}\|x\|_p^{-1})\leq
1-(\mathcal{F}J)(p\|x\|_p^{-1})$. This fact and Remark
\ref{obs_J}-(i) implies that
\begin{equation} \label{desi_1-J}
0\leq
e^{-t(1-(\mathcal{F}J)(p^{-j}\|x\|_p^{-1}))} - e^{-t(1-(\mathcal{F}J)(p\|x\|_p^{-1}))} \leq 1.
\end{equation}
Therefore, by (\ref{exp2 Z(x,t)}) we have that $Z_t(x)\geq 0$.
\item By Lemma \ref{Lemma_J}-(i) and Lemma \ref{TF_Z(t,x)} we have
that $(\mathcal{F}Z_t(x))(0)=1$ and
$(\mathcal{F}Z_t(x))(\xi)=\int_{\mathbf{\mathbb{Q}}_p^n}\chi_p\left(
\xi \cdot x\right)Z_t(x)d^nx$. So that
$(\mathcal{F}Z_t(x))(0)=\int_{\mathbf{\mathbb{Q}}_p^n}Z_t(x)d^nx$.
Therefore, $\int_{\mathbf{\mathbb{Q}}_p^n}Z_t(x)d^nx=1$.
\item Let $\varphi \in \mathcal{D}(\mathbb{Q}_p^n)$. Then by
\cite[Section 4.4]{Alberio et al} and using Fubini's we have that
\begin{align*}
\left\langle Z_{t+s}(x),\varphi\right\rangle
&=\int_{\mathbf{\mathbb{Q}}_p^n}Z_{t+s}(x)\varphi(x)d^nx\\
&=\int_{\mathbf{\mathbb{Q}}_p^n}\left[\int_{\mathbf{\mathbb{Q}}_p^n}\chi_p\left(
-x \cdot \xi\right)e^{-(t+s)(1-(\mathcal{F}J)(\|\xi\|_p))}d^n\xi\right]\varphi(x)d^nx \\
&=\int_{\mathbf{\mathbb{Q}}_p^n}\left(Z_t\ast Z_{s}\right)(x)\varphi(x)d^nx \\
&=\left\langle Z_t\ast Z_{s},\varphi \right\rangle.
\end{align*}
\end{enumerate}
\end{proof}

\begin{remark}
The fact that $1-(\mathcal{F}J)(\|\xi \|_p)$ is a increasing
function with respect to $\|\cdot\|_p$ marks a difference between
our heat Kernel and the heat Kernel treated in \cite{To-Z}, since by
(\ref{exp2 Z(x,t)}), (\ref{desi_1-J}) and that
$(1-p^{-n})\sum_{j=0}^{\infty}p^{-nj}=1$ we have that $Z_t(x)\leq
\|x\|_p^{-1}$, for $x\in
\mathbf{\mathbb{Q}}_p^n\backslash\{0\}$ and $t>0$. Moreover, to
know explicitly our function $Z_t(x)$, without the need to use an
auxiliary function $\widetilde{Z}(x,t)$ as in \cite{To-Z}.
\end{remark}

\begin{theorem}\label{theorem_classical_solution}
\begin{enumerate}[(i)]
\item The classical solution $u(x,t)$, $x\in \mathbf{\mathbb{Q}}$,  $t\geq 0$ given in (\ref{classical_solution}), satisfies
\begin{equation*}
\int_{\mathbf{\mathbb{Q}}_p^n} u(x,t) d^nx = \int_{\mathbf{\mathbb{Q}}_p^n} u_0(x) d^nx,
\end{equation*}
where $u_0(x)$ is the initial condition of the Cauchy problem
(\ref{Cauchy_problem}).
\item (Comparison Principle) Let $u(x,t)$ and $v(x,t)$, $x\in \mathbf{\mathbb{Q}}_p^n$,  $t\geq
0$, be fundamental solutions of the Cauchy problem
(\ref{Cauchy_problem}) with initial data $u_0$ and $v_0$
respectively. If $u_0(x)\geq v_0(x)$ for all $x\in
\mathbf{\mathbb{Q}}_p^n$, then
\begin{equation*}
u(x,t)\geq v(x,t) \ \  \text{for all} \ \ (x,t)\in
\mathbf{\mathbb{Q}}_p^n\times [0,\infty).
\end{equation*}
\item If the initial condition $u_0(x)$ of the Cauchy problem
(\ref{Cauchy_problem}) is a radial function such that $supp(\mathcal{F}u_0)=B_{M}^n$, $M\in \mathbb{Z}$, then $u(x,t)$ it's exactly\\
$-\left(1-\Omega \left( p^{M}\|x\|_p\right)\right)\|x\|^{-n}e^{-t(1-(\mathcal{F}J)(p\|x\|_p^{-1}))}\left(\mathcal{F}u_0\right)(p\|x\|_p^{-1})+$\\
$+(1-p^{-n})\sum_{j=K}^{\infty}p^{-nj}e^{-t(1-(\mathcal{F}J)(p^{-j}))}\left(\mathcal{F}u_0\right)(p^{-j}),$
where
\begin{equation*}
K=
\begin{cases}
-M, & \text{if}  \text{ }\|x\|_p\leq p^{-M}, \\
-M+1, & \text{if }  \text{\ }\|x\|_p=p^{-M+1}, \\
-ord_p(x), & \text{if}  \text{\ }\|x\|_p\geq p^{-M+2}.
\end{cases}
\end{equation*}
\end{enumerate}
\end{theorem}

\begin{proof}
\begin{enumerate}[(i)]
\item By using Fubini's Theorem and Theorem \ref{Z(t,x)_properties}-(i),
we have that
\begin{align*}
\int_{\mathbf{\mathbb{Q}}_p^n}u(x,t)d^nx
&=\int_{\mathbf{\mathbb{Q}}_p^n}\int_{\mathbf{\mathbb{Q}}_p^n}Z_t(x-y) u_0(y)d^nyd^nx\\
&=\int_{\mathbf{\mathbb{Q}}_p^n}u_0(y)\int_{\mathbf{\mathbb{Q}}_p^n} Z_t(x-y)d^nxd^ny\\
&=\int_{\mathbf{\mathbb{Q}}_p^n}u_0(y)d^ny.
\end{align*}
\item Following (\ref{classical_solution}), Theorem \ref{Z(t,x)_properties} and using Fubini's Theorem we have that
\begin{align*}
u(x,t)-v(x,t) &=\int_{\mathbf{\mathbb{Q}}_p^n}\chi_p(-x\cdot \xi) e^{-t(1-(\mathcal{F}J)(\|\xi\|_p))}\left\{(\mathcal{F}u_0)(\xi
)-(\mathcal{F}v_0)(\xi)\right\}d^n\xi \\
& = \int_{\mathbf{\mathbb{Q}}_p^n} (u_0(y)-v_0(y)) \left\{\int_{\mathbf{\mathbb{Q}}_p^n}\chi_p (-(x-y)\cdot \xi) e^{-t(1-(\mathcal{F}J)(\|\xi\|_p))}d^n\xi \right\}d^ny\\
&=\int_{\mathbf{\mathbb{Q}}_p^n}\left(u_0(y)-v_0(y)\right)Z_t(x-y)d^ny\geq 0.
\end{align*}
\item By (\ref{classical_solution}) and the fact that $supp(\mathcal{F}u_0)=B_{M}^n$, $M\in \mathbb{Z}$, we have that
\begin{align*}
u(x,t)&=\sum_{j=-\infty}^{M}e^{-t(1-(\mathcal{F}J))(p^j)}(\mathcal{F}u_0)(p^j)\int_{\|p^j\xi\|_p=1}\chi_p\left(
-x\cdot \xi \right)d^n\xi \\
&=\sum_{j=-\infty}^{M}p^{nj}e^{-t(1-(\mathcal{F}J))(p^j)}(\mathcal{F}u_0)(p^j)\int_{\|w\|_p=1}\chi_p\left(
-xp^{-j}\cdot w \right)d^nw \\
&=\sum_{j=-M}^{\infty}p^{-nj}e^{-t(1-(\mathcal{F}J))(p^{-j})}(\mathcal{F}u_0)(p^{-j})\int_{\|w\|_p=1}\chi_p\left(
-xp^j\cdot w \right)d^nw
\end{align*}
Taking into account this expansion of $u(x,t)$ and the formula
\begin{equation*}
\int_{\|w\|_p=1}\chi_p\left(-xp^j\cdot w\right) d^nw =
\begin{cases}
1-p^{-n}, & \text{if}  \text{ }\|x\|_p\leq p^j, \\
-p^{-n}, & \text{if }  \text{\ }\|x\|_p= p^{j+1}, \\
0, & \text{if}  \text{\ }\|x\|_p\geq p^{j+2},
\end{cases}
\end{equation*}
we will consider the following cases:
\begin{case}
$\|x\|_p\leq p^{-M}.$
\end{case}
In this case we have that
\begin{equation*}
u(x,t)=(1-p^{-n})\sum_{j=-M}^{\infty}p^{-nj}e^{-t(1-(\mathcal{F}J)(p^{-j}))}\left(\mathcal{F}u_0\right)(p^{-j}).
\end{equation*}
\begin{case}
$\|x\|_p=p^{-M+1}.$
\end{case}
For this case it is satisfied that
\begin{align*}
u(x,t) & = -p^{-n(1-M)}e^{-t(1-(\mathcal{F}J)(p^{M}))}(\mathcal{F}u_0)(p^{M}) + \\ & (1-p^{-n})\sum_{j=-M+1}^{\infty}p^{-nj}e^{-t(1-(\mathcal{F}J)(p^{-j}))} (\mathcal{F}u_0)(p^{-j}).
\end{align*}

\begin{case}
$\|x\|_p=p^{-L}\geq p^{-M+2}$, where $L=ord_p(x)$.
\end{case}
In this case we have that
\begin{align*}
u(x,t) & = -p^{nL}e^{-t(1-(\mathcal{F}J)(p^{L+1}))}\left(\mathcal{F}u_0\right)(p^{L+1}) +\\ &(1-p^{-n})\sum_{j=-L}^{\infty}p^{-nj}e^{-t(1-(\mathcal{F}J)(p^{-j}))}\left(\mathcal{F}u_0\right)(p^{-j}).
\end{align*}

From all the above, we obtain the desired equality.
\end{enumerate}
\end{proof}

\begin{remark}
The solutions of Cauchy problem (\ref{Cauchy_problem}) depend continuously on the initial data, in the following sense: if $u$ and $v$
are fundamental solutions of (\ref{Cauchy_problem}) with initial data $u_0$ and $v_0$ respectively, then
\begin{equation*}
\|u(x,t)-v(x,t)\|_{L^1(\mathbf{\mathbb{Q}}_p^n)}\leq\|u_0-v_0\|_{L^1(\mathbf{\mathbb{Q}}_p^n)}.
\end{equation*}
\end{remark}

\begin{definition}
\label{convolution_semigroup} A family $(\mu _t)_{t>0}$ of positive bounded
measures on $\mathbb{Q}_p^n$ with the properties
\begin{enumerate}[(i)]
\item $\mu _t(\mathbb{Q}_p^n)\leq 1$ for $t>0,$
\item $\mu _t\ast \mu _{s}=\mu _{t+s}$ for $t,s>0,$
\item $\lim_{t\rightarrow 0^{+}}\mu_t=\delta _0$ vaguely
($\delta _0 $ denotes the Dirac measure at $0\in
\mathbb{Q}_p^n)$,
\end{enumerate}
is called a convolution semigroup on $\mathbb{Q}_p^n.$
\end{definition}

\begin{theorem} \label{conv_semigroups}
The family $(Z(t,\cdot))_{t>0}$ determines a convolution semigroup on $\mathbb{Q}_p^n$.
\end{theorem}

\begin{proof}
Since $1-(\mathcal{F}J)$ is a continuous and negative definite function, see Lemma \ref{Lemma_J} and Remark \ref{obs_J}, we have
by \cite[Section \S 8]{Berg-Gunnar} that the family $\left(Z(t,\cdot)\right)_{t>0}$ corresponds to positive and bounded measures. Moreover,
by Lemma \ref{TF_Z(t,x)} is obtained that $\lim_{t\rightarrow 0} (\mathcal{F}Z_t(x))(\xi )=1$, so that by \cite[Section \S 10]{Berg-Gunnar} it is true that
$\lim_{t\rightarrow 0^{+}}Z(t,\cdot)=\delta _0$ vaguely.
From the above and the Theorem \ref{Z(t,x)_properties} we have
that the family $\left(Z(t,\cdot)\right)_{t>0}$, determines a convolution semigroup on $\mathbb{Q}_p^n.$
\end{proof}

\section{\label{Markov Processes}Strong Markov Processes, Linear and Logarithmic Energy Landscape, the First Passage Time Problem and the Probability of Survival}

In this section, we will study the first passage time problem associated with strong Markov processes (L\'evy processes) with transition functions of the types linear and logarithmic. Later, we will study the probability of survival for this transition functions. For the basic results on L\'evy
and strong Markov processes the reader may consult \cite{Blumenthal-Getoor}, \cite{Evans}.

\begin{definition} \label{Landscapes type definition}
\begin{enumerate}[(i)]
\item We say that function $J:$\textbf{\ }$\mathbb{Q}_p^n\rightarrow \mathbb{R}_{+}$ is of linear type if it satisfies the
Hypothesis A and there exist positive real constants $B$, $C$, $\alpha ,$
with $B\leq C,$ such that
\begin{equation*}
\frac{B}{\|x\|_p^{\alpha +1}}\leq J(\|x\|_p)\leq \frac{C}{
\|x\|_p^{\alpha +1}}\text{, for any }x\in\mathbb{Q}_p^n.
\end{equation*}

\item We say that function $J:$\textbf{\ }$\mathbb{Q}_p^n\rightarrow \mathbb{R}_{+}$ is of logarithmic type if it satisfies
the Hypothesis A and there exist positive real constants $D$, $E$, $\alpha$, $\beta$, with $D\leq E$ such that
\begin{equation*}
\frac{D\ln ^{\alpha }(1+\|x\|_p)}{\|x\|_p^{\beta }}\leq J(\|x\|_p)\leq
\frac{E\ln ^{\alpha }(1+\|x\|_p)}{\|x\|_p^{\beta }},
\end{equation*}
for any $x\in\mathbb{Q}_p^n.$
\end{enumerate}
\end{definition}

\begin{remark} \label{Remark J}
\begin{enumerate}[(i)]
\item For instance, in the one dimensional case the function $J(|x|_p)=\frac{1}{|x|_p\ln ^{\alpha
}(1+|x|_p)},$ $\alpha >1,$ which was used in \cite{Av-4} is not
integrable, indeed,
\[\int_{\mathbb{Q}_p} \frac{1}{|x|_p\ln^\alpha (1+|x|_p)} dx = (1-p^{-1}) \sum_{j=0}^{\infty }\frac{1}{\ln ^{\alpha }(1+p^{-j})}
+(1-p^{-1})\sum_{j=1}^{\infty }\frac{1}{\ln ^{\alpha }(1+p^j)}.\]
Note that
\begin{align*}
\ln ^{\alpha }(1+p^{-j}) &\leq (1+p^{-j})^{\alpha } \\
&\Leftrightarrow \frac{(1-p^{-1})}{(1+p^{-j})^{\alpha }}\leq \frac{(1-p^{-1})}{\ln ^{\alpha }(1+p^{-j})},
\end{align*}
moreover,
\begin{equation*}
(1-p^{-1})\lim_{j\rightarrow \infty }\frac{1}{(1+p^{-j})^{\alpha }}
=(1-p^{-1})\neq 0,
\end{equation*}
so that by the series divergence criterion $\sum_{j=0}^{\infty }
\frac{(1-p^{-1})}{(1+p^{-j})^{\alpha }}\rightarrow \infty $ and consequently
by the series comparison criterion, $(1-p^{-1})\sum_{j=0}^{\infty }
\frac{1}{\ln ^{\alpha }(1+p^{-j})}\rightarrow \infty .$ In general, we have
that the function $J(\|x\|_p)=\frac{B}{\|x\|_p\ln ^{\alpha }(1+\|x\|_p)
},$ $\alpha >1,$ $B>0,$ is not integrable in $\mathbb{Q}_p^n.$ This situation causes serious mathematical problems to achieve
the goals set in this paper. Hence the forms of the function $J$ of the logarithmic type\ in Definition \ref{Landscapes type definition}-(ii).

\item By Lemma \ref{lemma1}$-(ii)$ we have that
\begin{align*}
1-(\mathcal{F}J)(1)&=p^nJ(p)+(1-p^{-n})\sum_{j=2}^{\infty
}p^{nj}J(p^j)\\
&=p^nJ(p)+\int_{\mathbf{\mathbb{Q}}_p^n\setminus B_1^n}J(\|x\|_p)d^nx.
\end{align*}
If $1-(\mathcal{F}J)(1)=0$, then $J(p)=0$ and $J(\|x\|_p)=0$ for all $x\in\mathbb{Q}_p^n\backslash B_1^n$, i.e. $supp(J)\subseteq B_1^n$, which
is not possible since $J$ is of linear and logarithmic type. Therefore, $1-(\mathcal{F}J)(1)>0$.
\end{enumerate}
\end{remark}

\begin{example}
\begin{enumerate}[(i)]
\item Let $\alpha >0$ for which there exists $r\in\mathbb{Q}\backslash \{0\}$ such that $\alpha +1=n+r,$ and for this $r$ there exists $
N:=N(r)\in\mathbb{N}$ with the property that if\ $j\in \mathbb{N}$ and $j\geq N$ then $j+\frac{n}{r},$ $j-\frac{n}{r}\in \mathbb{N}$. Then
the function $J(\|x\|_p):=\frac{B}{\|x\|^{\alpha +1}},$ $B>0$, is
of the linear type where $B$ is a constant so that $\int_{\mathbb{Q}_p^n}J(\|x\|_p)d^nx=1$. For example, the function $J(\|x\|_p)=
\frac{B}{\|x\|^{\alpha +1}}$ with $\alpha +1=n+r,$ where $n=3$ and $r=\frac{1
}{2}$ is of the linear type.

\item The functions $J(\|x\|_p):=\frac{B\ln ^{\alpha }(1+\|x\|_p)}{
\|x\|^{\alpha }}$ for any $\alpha >0$ are of logarithmic type, here $B$ is a
positive constant so that $\int_{\mathbb{Q}_p^n}J(\|x\|_p)d^nx=1$.
\end{enumerate}
\end{example}

The next lemma offers us conditions for determine when the function $\frac{1}{1- (\mathcal{F}J)}$ is not integrable when we restrict it to the unit ball.

\begin{lemma} \label{1-TFJ int_no}
Suppose that the function $J:$\textbf{\ }$\mathbb{Q}_p^n\rightarrow \mathbb{R}_{+}$ satisfies one of the following
conditions:
\begin{enumerate}[(i)]
\item $J$ is of linear type with $\alpha +1-2n\geq 0$,
\item $J$ is of logarithmic type with $\beta -\alpha -2n\geq 0$,
\end{enumerate}
Then
\begin{equation*}
\frac{\Omega \left( \|\xi \|_p\right) }{1- (\mathcal{F}J)(\|\xi \|_p)}\notin
L^1\left(\mathbb{Q}_p^n,d^n\xi \right).
\end{equation*}
\end{lemma}

\begin{proof}
\begin{enumerate}[(i)]
\item By Lemma \ref{lemma1}$-(ii)$ and Definition \ref{Landscapes type definition}-(i), we have for $\xi \in \mathbb{Q}_p^n\backslash \{0\}$,
\begin{align*}
1-(\mathcal{F}J)(\|\xi \|_p)&\leq \|\xi\|_p^{-n}\left\{\frac{Cp^n}{p^{\alpha+1}\|\xi\|_p^{-(\alpha+1)}}+(1-p^{-n})\sum_{j=2}^{\infty}\frac{Cp^{nj}}{p^{j(\alpha+1)}\|\xi\|_p^{-(\alpha+1)}}\right\}\\
&=C\|\xi\|_p^{\alpha+1-n}\left\{p^{n-\alpha-1}+(1-p^{-n})\sum_{j=2}^{\infty}p^{j(n-\alpha-1)}\right\}.
\end{align*}
Note that the condition\ $\alpha +1-2n\geq 0$ implies that $n-\alpha-1<0$, so
the series $\sum_{j=2}^{\infty }p^{j(n-\alpha-1)}<\infty .$ Then,
there exists a positive constant $C_2$ such that
\begin{equation}
1-(\mathcal{F}J)(\|\xi \|_p)\leq C_2 \|\xi \|_p^{\alpha+1-n},\text{ }\xi
\in \mathbb{Q}_p^n\backslash \{0\}.  \label{1-TF(J)<=  linear}
\end{equation}
Therefore,
\begin{align*}
\int_{\mathbb{Z}_p^n}\frac{d^n\xi }{1-(\mathcal{F}J) \left( \|\xi
\|_p\right) }
&\geq \frac{1}{C_{2}}\int_{\mathbb{Z}_p^n}
\frac{d^n\xi }{\|\xi \|_p^{\alpha +1-n}}\\
&= \frac{1}{C_{2}}
\sum\limits_{j=0}^{\infty }\frac{1}{p^{-j(\alpha +1-n)}}\int\limits_{\|\xi
\|_p=p^{-j}}d^n\xi \\
&=\frac{(1-p^{-n})}{C_{2}}\sum\limits_{j=0}^{\infty }p^{j(\alpha
+1-2n)}\rightarrow \infty .
\end{align*}

\item By Lemma \ref{lemma1}$-(ii)$, Definition \ref{Landscapes type definition}-(ii) and applying the inequality $\ln^{\alpha}
(1+p^j\|\xi\|_p^{-1})\leq p^{j\alpha}\|\xi\|_p^{-\alpha}$, $j\geq 1$, we have that $1- (\mathcal{F}J)(\|\xi \|_p),$ $\xi \in  \mathbb{Q}_p^n\backslash \{0\}$,
is dominated superiorly by
\begin{align*}
&\|\xi\|_p^{-n}\left\{\frac{Ep^nln^{\alpha}(1+p\|\xi\|_p^{-1})}{p^{\beta}\|\xi\|_p^{-\beta}}+(1-p^{-n})\sum_{j=2}^{\infty}\frac{Eln^{\alpha}(1+p^j\|\xi\|_p^{-1})p^{nj}}{p^{j\beta}\|\xi\|_p^{-\beta}}\right\}\\
&\leq E\|\xi\|_p^{\beta-n-\alpha} \left\{p^{n+\alpha-\beta} + (1-p^{-n}) \sum_{j=2}^{\infty}p^{j(\alpha+n-\beta)}\right\}.
\end{align*}

Note that the condition\ $\beta -\alpha -2n\geq 0$ implies that $\alpha+n-\beta
<0$, so the series $\sum_{j=2}^{\infty }p^{j(\alpha+n-\beta)}<\infty .$ Therefore, there exists a positive constant $E_{2}$ such that
\begin{equation}
1-(\mathcal{F}J)(\|\xi \|_p)\leq E_{2}\|\xi \|_p^{\beta -n-\alpha },\text{ }
\xi \in \mathbb{Q}_p^n\backslash \{0\}.  \label{1-TF(J)<=  logaritmic}
\end{equation}
Consequently,
\begin{align*}
\int_{\mathbb{Z}_p^n}\frac{d^n\xi }{1- (\mathcal{F}J)\left( \|\xi
\|_p\right) } &\geq \frac{1}{E_{2}}\int_{\mathbb{Z}_p^n}
\frac{d^n\xi }{\|\xi \|_p^{\beta -n-\alpha }}\\
&=\frac{1}{E_{2}}
\sum\limits_{j=0}^{\infty }\frac{1}{p^{-j(\beta -n-\alpha )}}
\int\limits_{\|\xi \|_p=p^{-j}}d^n\xi \\
&=\frac{(1-p^{-n})}{E_{2}}\sum\limits_{j=0}^{\infty }p^{j(\beta -\alpha
-2n)}\rightarrow \infty .
\end{align*}
\end{enumerate}
\end{proof}

\begin{remark} \label{remarks on dif type}
By Lemma \ref{Lemma_J}-$(i)$, we have that $1- (\mathcal{F}J)(0)=0$. On the other hand, by Lemma \ref{Lemma_J}-$(ii)$ and
Definition \ref{Landscapes type definition} the following considerations are obtained for $\xi \in \mathbb{Q}_p^n \backslash \{0\}$:
\begin{enumerate}[(i)]
\item If $J$ is of linear type with $J(\|x\|_p) \approx \frac{C}{
\|x\|_p^{\alpha +1}}$, then there is a constant $F>0$ such that $$1-(\mathcal{F}J)(\|\xi\|_p)=F\|\xi\|_p^{\alpha+1-n}.$$
In this case we have that $1-(\mathcal{F}J)(\|\xi\|_p)$ is a increasing function with
respect to $\|\cdot\|_p$.
\item If $J$ is of logarithmic type with $J(\|x\|_p) \approx \frac{E\ln ^{\alpha }(1+\|x\|_p)}{\|x\|_p^{\beta }}$,
then $1-(\mathcal{F}J)(\|\xi\|_p)$ it's exactly
$$E\|\xi\|_p^{\beta-n}\left\{p^{n-\beta}ln^{\alpha}(1+p\|\xi\|_p^{-1})+(1-p^{-n})\sum_{j=2}^{\infty}p^{j(n-\beta)}ln^{\alpha}(1+p^j\|\xi\|_p^{-1})\right\}.$$
In this case we have that the functions $ln^{\alpha}(1+p\|\xi\|_p^{-1})$ and $ln^{\alpha}(1+p^j\|\xi\|_p^{-1})$, $j\geq 2$, are
decreasing with respect to $\|\cdot\|_p$, while the function $E\|\xi\|_p^{\beta-n}$ is a increasing function with
respect to $\|\cdot\|_p$. So suitably taking $\alpha$, we have that $1-(\mathcal{F}J)(\|\xi\|_p)$ is a increasing function with
respect to $\|\cdot\|_p$.
\end{enumerate}
\end{remark}

We now consider the Cauchy problem
\begin{equation}
\left\{
\begin{array}{ll}
\frac{\partial u}{\partial t}(x,t)=\mathcal{A}u(x,t), & t\in \left[
0,\infty \right) ,\text{\ }x\in \mathbb{Q}_p^n \\
&  \\
u(x,0)=\Omega (\|x\|_p). &
\end{array}
\right.  \label{Cauchy_problem_2}
\end{equation}
by Remark \ref{Obs_operator} it is true that \begin{equation}
u(x,t)=Z_t(x)\ast \Omega (\|x\|_p) = \int_{\mathbf{\mathbb{Q}}_p^n}\chi _p\left( -\xi \cdot x\right) e^{-t(1- (\mathcal{F}J)(\|\xi \|_p))}  \Omega (\|x\|_p)d^n\xi ,  \label{u(x,t)}
\end{equation}%
is a classical solution of the Cauchy problem (\ref{Cauchy_problem_2}). Moreover, proceeding in the same way as in \cite{To-Z} we
have that there exists a strong Markov processes (more exactly a L\'evy process)
$\mathfrak{J}(t,\omega )$ with state space
$(\mathbb{Q}_p^n,\mathcal{B}\left(\mathbb{Q}_p^n\right))$
and transition function $q_t(x,\cdot )$ given by
\begin{equation*}
q_t(x,E)=\left\{
\begin{array}{ll}
Z_t(x)\ast 1_{E}(x)\text{,} & \text{\ for }t>0\text{, }x\in
\mathbf{\mathbb{Q}}_p^n,\text{ }E\in \mathcal{B}\left(\mathbb{Q}_p^n\right) \text{ } \\
&  \\
1_{E}(x), & \text{for }t=0\text{, }x\in
\mathbf{\mathbb{Q}}_p^n,\text{ }E\in
\mathcal{B}\left(\mathbb{Q}_p^n\right).
\end{array}
\right.
\end{equation*}
Set $\Upsilon $ to be the space of all paths $\mathfrak{J}(t,\omega )$. It is satisfied that all trajectories of the
process $\mathfrak{J}(t,\omega )$ start in $\mathbb{Z}_p^n$.

\begin{definition}
The random variable $\tau _{\mathbb{Z}_p^n}(\omega ):$ $\Upsilon \longrightarrow \mathbb{R}_{+}\cup \{+\infty \}$ defined by the relation
\begin{equation*}
\inf \{t>0:\mathfrak{J(}t,\omega \mathfrak{)}\in\mathbb{Z}_p^n\text{ and there exists }t^{\prime }\text{ such that }0<t^{\prime }<t
\text{ and }\mathfrak{J(}t^{\prime },\omega \mathfrak{)\notin }\mathbb{Z}_p^n\}
\end{equation*}%
is called the first passage time of a trajectory of the random process $
\mathfrak{J(}t,\omega \mathfrak{)}$ entering the domain $\mathbb{Z}_p^n$ (i.e., the first instant when it returns to $\mathbb{Z}_p^n$ ).
\end{definition}

\begin{definition}
We say that $\mathfrak{J(}t,\omega \mathfrak{)}$ is recurrent with respect
to $\mathbb{Z}_p^n$ if
\begin{equation}
P(\{\omega \in \Upsilon :\tau _{\mathbb{Z}_p^n}(\omega )\mathfrak{<}\infty \})=1,  \label{Probability}
\end{equation}
i.e. every path of $\mathfrak{J(}t,\omega \mathfrak{)}$ is sure to return to
$\mathbb{Z}_p^n$. If (\ref{Probability}) does not hold then we say that $\mathfrak{J(}t,\omega \mathfrak{)}$ is transient with
respect to $\mathbb{Z}_p^n,$ in this case, there exist paths of $\mathfrak{J(}t,\omega
\mathfrak{)}$ that abandon $\mathbb{Z}_p^n$ and never go back.
\end{definition}

\begin{lemma} \cite[Lemma 8]{To-Z} \label{relation g and f} The probability density function $f(t)$ of the
random variable $\tau _{\mathbb{Z}_p^n}(\omega )$ satisfies the non-homogeneous Volterra equation of
second kind
\begin{equation}
g(t)=\int\limits_0^{\infty }g(t-\tau )f(\tau )d\tau +f(t)
\label{Volterra_equ}
\end{equation}
where
\begin{equation*}
g(t)=\int_{\mathbf{\mathbb{Q}}_p^n\backslash
\mathbb{Z}_p^n}J(\|y\|_p)u(y,t)d^ny,  \label{g(t)}
\end{equation*}
is the probability density function for a path of $\mathfrak{J(}t,\omega
\mathfrak{)}$ to enter into $\mathbb{Z}_p^n$ at the instant of time $t$, with the condition that $\mathfrak{J(}
0,\omega \mathfrak{)\in }\mathbb{Z}_p^n.$
\end{lemma}

\begin{remark} \label{Laplace}
Applying Laplace transform at (\ref{Volterra_equ}) we have that
\begin{equation*}
F(s)=1-\frac{1}{1+G(s)},
\end{equation*}
where $F(s)$ and $G(s)$ are the Laplace transforms of $f$ and $g$, respectively.

If $f(t)$ is the probability density function of the
random variable $\tau _{\mathbb{Z}_p^n}(\omega )$, then $f(t)$ is the probability that a trajectory will enter
or return for the first time to the unit ball. So $F(0)=\int_0^{\infty}f(t)dt$ is the probability that the trajectory will ever return to the unit ball.
$F(0)$ can be calculated, because it is to know who is $G(0)$. In this case, $F(0)=1-\frac{G(0)}{1+G(0)}$.

Then, if $G(0)$ exists, it happens that $F(0)<1$ and this fact indicates that there are trajectories that do not necessarily return
to $\mathbb{Z}_p^n$. Conversely, if $G(0)=\infty$ then $F(0)=1$, and this means that with probability $1$ the trajectory returns to $\mathbb{Z}_p^n$.
\end{remark}

By Lemma \ref{1-TFJ int_no}, Remark \ref{Laplace} and proceeding analogously as in \cite[Theorem 3]{To-Z}, we obtain the following theorem.

\begin{theorem}
\label{Theorem3} If $J$ satisfies one of the hypothesis of Lemma
\ref{1-TFJ int_no}, then $\mathfrak{J}(t,\omega )$ is recurrent with
respect to $\mathbb{Z}_p^n$.
\end{theorem}

\begin{remark}
The survival probability $S(t)$ (the probability that a path of $\mathfrak{J(
}t,\omega \mathfrak{)}$ remains in $\mathbb{Z}_p^n$ at the time $t$) is given by
\begin{equation*}
S(t):=S_{\mathbb{Z}_p^n}(t)=\int_{\mathbf{\mathbb{Z}}_p^n}u(x,t)d^nx,  \label{S(t)}
\end{equation*}
where $u(x,t)$ is given by (\ref{u(x,t)}). Therefore,
\begin{eqnarray}
S(t) &=&\int_{\mathbf{\mathbb{Z}}_p^n}\int_{\mathbf{\mathbb{Z}}_p^n}\chi _p\left( -\xi \cdot x\right) e^{-t(1-(\mathcal{F}J)(\|\xi\|_p))}d^n\xi  \notag \\
&=&\int_{\mathbf{\mathbb{Z}}_p^n}e^{-t(1-(\mathcal{F}J)(\|\xi \|_p))}d^n\xi
=(1-p^{-n})\sum_{j=0}^{\infty}p^{-nj}e^{-t(1-(\mathcal{F}J)(p^{-j}))}.
\label{S(t) expansion}
\end{eqnarray}
\end{remark}

In a part of the proof of the following theorem we will use the technique of \cite[Appendix A]{Av-2}. In addition, the
inequalities obtained here have different form to those obtained in \cite{Av-2}.

\begin{theorem} \label{Theorem last}
\begin{enumerate}[(i)]
\item If $J$ is of linear type and satisfies the
hypothesis of Lemma \ref{1-TFJ int_no}$-(i)$ then there exist positive
real constants $B_{3}=\frac{p^n-1}{(\alpha +1-n)\ln p}$, $A_{3}=\frac{p^n-1}{(\alpha+1-n)\ln p}$, such that
\begin{equation}
\frac{B_{3}}{(t C_{2})^{\frac{n}{\alpha +1-n}}}\gamma \left( \frac{n}{\alpha
+1-n},t C_{2}\right) \leq S(t)\leq \frac{A_{3}}{(tA_1)^{\frac{n}{\alpha +1-n
}}}\gamma \left( \frac{n}{\alpha +1-n},tA_1\right) ,  \label{Lemma (i)}
\end{equation}
where $\gamma (s,x):=\int_0^{x}z^{s-1}e^{-z}dz$ is the lower incomplete
gamma function, $C_{2}$ is the constant given in (\ref{1-TF(J)<= linear})
and $A_1=\frac{A}{p^{\alpha +1}}.$

\item If $J$ is of logarithmic type and satisfies the hypothesis of
Lemma \ref{1-TFJ int_no}$-(ii)$ then there exist positive real
constants $B_{4}=\frac{p^n-1}{(\beta-n-\alpha )\ln p}$, $A_{4}=\frac{p^n-1}{
(\beta-n)\ln p}$, such that
\begin{equation*}
\frac{B_{4}}{(tE_{2})^{\frac{n}{\beta-n-\alpha}}}\gamma \left( \frac{n}{
\beta-n-\alpha},tE_{2}\right) \leq S(t)\leq \frac{A_{4}}{(tE_{3})^{\frac{n}{\beta-n}}}\gamma \left( \frac{n}{\beta-n},tE_{3}\right) ,
\end{equation*}
where $\gamma (s,x):=\int_0^{x}z^{s-1}e^{-z}dz$ is the lower incomplete
gamma function, $E_{2}$ is the constant given in (\ref{1-TF(J)<= logaritmic}) and $E_{3}$ is the constant given in (\ref{des_I_4}).
\end{enumerate}
\end{theorem}

\begin{proof}
\begin{enumerate}[(i)]
\item From Lemma \ref{lemma1}$-(ii),$ and the fact that $J$ is of linear
type we get that
\begin{align*}
1-(\mathcal{F}J)\left(\|\xi \|_p\right)&\geq B\|\xi \|_p^{\alpha+1-n}\left\{p^{n-\alpha-1}+(1-p^{-n})\sum_{j=2}^{\infty}p^{j(n-\alpha-1)} \right\}.
\end{align*}
Since $n-\alpha-1<0$ we have that $\sum_{j=2}^{\infty}p^{j(n-\alpha-1)}<\infty$. Therefore, there exists a positive constant $A_1$ such that
\begin{equation*}
1-(\mathcal{F}J)\left( \|\xi \|_p\right) \geq A_1\|\xi \|_p^{\alpha +1-n},
\text{ }\xi \in \mathbf{\mathbb{Q}}_p^n\backslash \{0\}.
\end{equation*}
So that by (\ref{S(t) expansion}) we have that
\begin{equation}
S(t)\leq (1-p^{-n})\sum_{j=0}^{\infty }p^{-nj}e^{-tA_1p^{-j(\alpha +1-n)}}.
\label{S(t)<=sum}
\end{equation}
We know that $e^{-tA_1p^{-x(\alpha +1-n)}}$ is an increasing function and $
p^{-nx}$ is a decreasing function in the variable $x$. Therefore, we have on
the interval $j\leq x\leq j+1$ the inequalities
\begin{equation}
\frac{e^{-tA_1p^{-(x-1)(\alpha +1-n)}}}{p^{nx}}\leq \frac{
e^{-tA_1p^{-j(\alpha +1-n)}}}{p^{nj}}\leq \frac{e^{-tA_1p^{-x(\alpha
+1-n)}}}{p^{n(x-1)}}.  \label{e^-t<=}
\end{equation}
Integrating (\ref{e^-t<=}) in the variable $x$ from $j$ to $j+1,$ we get
\begin{equation*}
\int_j^{j+1}\frac{e^{-tA_1p^{-(x-1)(\alpha +1-n)}}}{p^{nx}}dx\leq
\int_j^{j+1}\frac{e^{-tA_1p^{-j(\alpha +1-n)}}}{p^{nj}}dx\leq
\int_j^{j+1}\frac{e^{-tA_1p^{-x(\alpha +1-n)}}}{p^{n(x-1)}}dx.
\end{equation*}
So that
\begin{equation*}
\int_j^{j+1}\frac{e^{-tA_1p^{-(x-1)(\alpha +1-n)}}}{p^{nx}}
dx\leq \frac{e^{-tA_1p^{-j(\alpha +1-n)}}}{p^{nj}}\leq \int_j^{j+1}\frac{
e^{-tA_1p^{-x(\alpha +1-n)}}}{p^{n(x-1)}}dx.
\end{equation*}%
Now, summing with respect to $j$ from $0$ to $\infty$ we have that
\begin{equation*}
\int_0^{\infty }\frac{e^{-tA_1p^{-(x-1)(\alpha +1-n)}}}{
p^{nx}}dx\leq \sum_{j=0}^{\infty }p^{-nj}e^{-tA_1p^{-j(\alpha
+1-n)}}\leq p^n\int_0^{\infty }\frac{e^{-tA_1p^{-x(\alpha +1-n)}}}{
p^{nx}}dx.
\end{equation*}
Next, we will consider only the inequality
\begin{equation} \label{des_I_1}
\sum_{j=0}^{\infty }p^{-nj}e^{-tA_1p^{-j(\alpha
+1-n)}}\leq p^n\int_0^{\infty }\frac{e^{-tA_1p^{-x(\alpha +1-n)}}}{
p^{nx}}dx=I_1.
\end{equation}
For $I_1$ let's change the variable $z=tA_1p^{-x(\alpha+1-n)}$. Then, we have that $dz=-z(\alpha+1-n)\ln p \text{ } dx$
and $x=\frac{log_p(z^{-1}tA_1)}{\alpha+1-n}$.\\
On the other hand, notice that when $x\rightarrow 0$ then $z\rightarrow tA_1$ and if $x\rightarrow \infty$ then $z\rightarrow 0$.
Therefore,
\begin{align*}
I_1&=\frac{p^n}{(\alpha+1-n)(tA_1)^{\frac{n}{\alpha+1-n}}\ln p}\int_0^{tA_1}e^{-z}z^{\frac{n}{\alpha+1-n}-1}dz\\
&=\frac{p^n}{(\alpha+1-n)(tA_1)^{\frac{n}{\alpha+1-n}}\ln p}\gamma \left( \frac{n}{\alpha +1-n},tA_1\right).
\end{align*}
So that by (\ref{S(t)<=sum}) and (\ref{des_I_1}) we have that
\begin{equation*}
S(t)\leq \frac{A_{3}}{(tA_1)^{\frac{n}{\alpha+1-n}}}
\gamma \left( \frac{n}{\alpha +1-n},tA_1\right).
\end{equation*}

On the other hand, by (\ref{1-TF(J)<=  linear}) and (\ref{S(t) expansion}) we have that
\begin{equation}
S(t)\geq (1-p^{-n})\sum_{j=0}^{\infty}p^{-nj}e^{-tC_{2}p^{-j(\alpha+1-n)}}.
\label{S(t)<=sum2}
\end{equation}
We know that $e^{-tC_{2}p^{-x(\alpha+1-n)}}$ is an increasing function and $
p^{-nx}$ is a decreasing function in the variable $x$. Therefore, we have on
the interval $j\leq x\leq j+1$ the inequalities
\begin{equation}
\frac{e^{-tC_{2}p^{-(x-1)(\alpha+1-n)}}}{p^{nx}}\geq \frac{
e^{-tC_{2}p^{-j(\alpha+1-n)}}}{p^{nj}}\geq \frac{e^{-tC_{2}p^{-x(\alpha
+1-n)}}}{p^{n(x-1)}}.  \label{e^-t<=2}
\end{equation}
Integrating (\ref{e^-t<=2}) in the variable $x$ from $j$ to $j+1,$ we get
\begin{equation*}
\int_j^{j+1}\frac{e^{-tC_{2}p^{-(x-1)(\alpha +1-n)}}}{p^{nx}}dx\geq
\int_j^{j+1}\frac{e^{-tC_{2}p^{-j(\alpha +1-n)}}}{p^{nj}}dx\geq
\int_j^{j+1}\frac{e^{-tC_{2}p^{-x(\alpha +1-n)}}}{p^{n(x-1)}}dx.
\end{equation*}
So that
\begin{equation*}
\int_j^{j+1}\frac{e^{-tC_{2}p^{-(x-1)(\alpha +1-n)}}}{p^{nx}}
dx\geq \frac{e^{-tC_{2}p^{-j(\alpha +1-n)}}}{p^{nj}}\geq \int_j^{j+1}\frac{
e^{-tC_{2}p^{-x(\alpha +1-n)}}}{p^{n(x-1)}}dx.
\end{equation*}
Now, summing with respect to $j$ from $0$ to $\infty$ we have that
\begin{equation*}
\int_0^{\infty }\frac{e^{-tC_{2}p^{-(x-1)(\alpha +1-n)}}}{
p^{nx}}dx\geq \sum_{j=0}^{\infty }p^{-nj}e^{-tC_{2}p^{-j(\alpha
+1-n)}}\geq p^n\int_0^{\infty }\frac{e^{-tC_{2}p^{-x(\alpha +1-n)}}}{
p^{nx}}dx.
\end{equation*}
Next, we will consider only the inequality
\begin{equation} \label{des_I_2}
\sum_{j=0}^{\infty }p^{-nj}e^{-tC_{2}p^{-j(\alpha
+1-n)}}\geq p^n\int_0^{\infty }\frac{e^{-tC_{2}p^{-x(\alpha +1-n)}}}{
p^{nx}}dx=I_{2}.
\end{equation}
For $I_{2}$ let's change the variable $z=tC_{2}p^{-x(\alpha+1-n)}$. Then, we have that $dz=-z(\alpha+1-n)\ln p \text{ } dx$
and $x=\frac{\log_p(z^{-1}tC_{2})}{\alpha+1-n}$.\\
On the other hand, notice that when $x\rightarrow 0$ then $z\rightarrow tC_{2}$ and if $x\rightarrow \infty$ then $z\rightarrow 0$.
Therefore,
\begin{align*}
I_{2}&=\frac{p^n}{(\alpha+1-n)(tC_{2})^{\frac{n}{\alpha+1-n}}\ln p}\int_0^{tC_{2}}e^{-z}z^{\frac{n}{\alpha+1-n}-1}dz\\
&=\frac{p^n}{(\alpha+1-n)(tC_{2})^{\frac{n}{\alpha+1-n}}\ln p}\gamma \left(\frac{n}{\alpha +1-n},tC_{2}\right).
\end{align*}
So that by (\ref{S(t)<=sum2}) and (\ref{des_I_2}) we have that
\begin{equation*}
S(t)\geq \frac{B_{3}}{(tC_{2})^{\frac{n}{\alpha+1-n}}}
\gamma \left( \frac{n}{\alpha +1-n},tC_{2}\right).
\end{equation*}
\item For $\xi\in \mathbb{Q}_p^n\setminus \{0\}$ we have by Lemma \ref{lemma1}$-(ii)$ and the fact that $J$ is of logarithmic
type that $1-(\mathcal{F}J)\left(\|\xi \|_p\right)$ exceeds
\begin{eqnarray*}
&& D\|\xi \|_p^{\beta-n}\left\{p^{n-\beta}\ln^{\alpha}(1+p\|\xi\|^{-1})+(1-p^{-n})\sum_{j=2}^{\infty}p^{j(n-\beta)}\ln^{\alpha}(1+p^j\|\xi\|_p^{-1})\right\}\\
&\geq& D\|\xi\|_p^{\beta-n-\alpha}\left\{\frac{p^{n+\alpha-\beta}}{(1+p\|\xi\|^{-1})^{\alpha}}+(1-p^{-n})\sum_{j=2}^{\infty}\frac{p^{j(n+\alpha-\beta)}}{(1+p^j\|\xi\|_p^{-1})^{\alpha}}\right\}\\
&\geq& \frac{Dp^{n+\alpha-\beta}\|\xi\|_p^{\beta-n-\alpha}}{(1+p\|\xi\|^{-1})^{\alpha}}\\
&=& \frac{Dp^{n+\alpha-\beta}\|\xi\|_p^{\beta-n}}{(\|\xi\|_p+p)^{\alpha}}.
\end{eqnarray*}
Therefore, if $J$ is of logarithmic type we get that
\begin{equation} \label{des_I_3}
1-(\mathcal{F}J)\left(\|\xi \|_p\right)\geq \frac{Dp^{n+\alpha-\beta}\|\xi\|_p^{\beta-n}}{(\|\xi\|_p+p)^{\alpha}}, \text{  for  } \xi\in \mathbf{\mathbb{Q}}_p^n\backslash \{0\}.
\end{equation}
Now, from (\ref{des_I_3}) we have that
\begin{equation} \label{des_I_4}
1-(\mathcal{F}J)\left(\|\xi \|_p\right)\geq E_{3}\|\xi\|_p^{\beta-n}, \text{  if  } \xi\in \mathbf{\mathbb{Z}}_p^n\backslash \{0\},
\end{equation}
where $E_{3}=\frac{Dp^{n+\alpha-\beta}}{(2+p)^{\alpha}}$.\\
Therefore, by (\ref{1-TF(J)<=  logaritmic}), (\ref{des_I_4}) and proceeding analogously as the case $(i)$, we obtain the desired inequalities.
\end{enumerate}
\end{proof}

\end{document}